\documentclass[10pt,journal,epsfig]{IEEEtran}

\usepackage[dvips]{graphicx}
\usepackage{graphicx}
\usepackage{amssymb}
\usepackage{cite}
\usepackage{subfigure}
\usepackage{amsmath}

\begin{document}

\title{Recovery of Block-Sparse Representations from Noisy Observations via Orthogonal Matching Pursuit}

\author{Jun Fang,~\IEEEmembership{Member,~IEEE}, and Hongbin Li,~\IEEEmembership{Senior Member,~IEEE}
\thanks{Jun Fang and Hongbin Li are with the Department of Electrical and Computer
Engineering, Stevens Institute of Technology, Hoboken, NJ 07030,
USA, E-mails: \{Jun.Fang@ieee.org, Hongbin.Li@stevens.edu\} }
\thanks{This work was supported in part by the National Science
Foundation under Grant ECCS-0901066 and by the Air Force Office of
Scientific Research under Grant FA9550-09-1-0310. }}

\maketitle

\begin{abstract}
We study the problem of recovering the sparsity pattern of
block-sparse signals from noise-corrupted measurements. A simple,
efficient recovery method, namely, a block-version of the
orthogonal matching pursuit (OMP) method, is considered in this
paper and its behavior for recovering the block-sparsity pattern
is analyzed. We provide sufficient conditions under which the
block-version of the OMP can successfully recover the block-sparse
representations in the presence of noise. Our analysis reveals
that exploiting block-sparsity can improve the recovery ability
and lead to a guaranteed recovery for a higher sparsity level.
Numerical results are presented to corroborate our theoretical
claim.
\end{abstract}



\section{Introduction}
The problem of recovering a high dimensional sparse signal based
on a small number of measurements has been of significant interest
in signal and imaging processing, applied mathematics, and
statistics. Such a problem arises from a number of applications,
including subset selection in regression \cite{Miller90},
structure estimation in graphical models \cite{Meinshausen06}, and
compressed sensing \cite{CandesTao05}. Among these applications,
many involves determining the locations of the nonzero components
of the sparse signal, which is also referred to as sparsity
pattern recovery (or more simply, sparsity recovery). In practice,
the locations of the nonzero components (or, the support of the
sparse signals) usually have significant physical meanings. For
example, in chemical agent detection, the indices for the nonzero
coordinates indicates the chemical components present in a
mixture. In sparse linear regression, the recovered support
corresponds to a small subset of features which linearly influence
the observed data. Due to its importance, sparsity pattern
recovery has received considerable attention over the past few
years.  In \cite{Fuchs05,Tropp06}, the authors analyzed the
behavior of $\ell_1$-constrained quadratic programming (QP), also
referred to as the Lasso, for recovering the sparsity pattern in a
deterministic framework. Sufficient conditions were established
for exact sparsity pattern recovery. Such a problem was also
studied in \cite{Wainwright09a} from a statistical perspective,
where necessary and sufficient conditions on the problem
dimension, the number of nonzero elements, and the number of
measurements are established for sparsity pattern recovery.
Recently, information-theoretic limits of sparsity recovery with
an exhaustive search decoder were studied in
\cite{Wainwright09,FletcherRangan09}.


In this paper, we consider the problem of recovering block-sparse
signals whose nonzero elements appear in fixed blocks.
Block-sparse signals arise naturally. For example, the atomic
decomposition of multi-band signals \cite{Mishali09} or audio
signals \cite{GribonvalBacry03} usually results in a block-sparse
structure in which the nonzero coefficients occur in clusters.
Recovery of block-sparse signals has been extensively studied in
\cite{EldarKuppinger10,EldarMishali09,BaraniukCevher10}, in which
the recovery behaviors of the basis pursuit (BP), or
$\ell_1$-constrained QP, and the orthogonal matching pursuit (OMP)
algorithms were analyzed via the restricted isometry property
(RIP) \cite{EldarMishali09,BaraniukCevher10} and the mutual
coherence property \cite{EldarKuppinger10}. Their analyses
\cite{EldarKuppinger10,EldarMishali09,BaraniukCevher10} revealed
that exploiting block-sparsity yields a relaxed condition which
can guarantee recovery for a higher sparsity level as compared
with treating block-sparse signals as conventional sparse signals.
Nevertheless, most of these studies focused on noiseless
scenarios. In practice, measurements are inevitably contaminated
with noise and underlying uncertainties. It is therefore important
to analyze the effect of measurement noise on the block-sparsity
pattern recovery, e.g. under what conditions the exact sparsity
pattern can be recovered, and does exploiting block-sparsity still
lead to a guaranteed recovery for a higher sparsity level? These
questions will be addressed in this paper. Specifically, we
consider a block version of the OMP algorithm and study its
behavior for recovering block-sparsity pattern in the presence of
noise. A comparison with the theoretical results for the
conventional OMP algorithm \cite{Tropp06} is presented to
highlight the benefits of exploiting block-sparsity property.





\section{Problem Formulation}
We consider the problem of recovering a block-sparse signal
$\mathbf{x}\in\mathbb{R}^{n}$ from noise-corrupted measurements
\begin{align}
\mathbf{y}=\mathbf{A}\mathbf{x}+\mathbf{w}
\end{align}
where $\mathbf{A}\in\mathbb{R}^{m\times n}$ ($m<n$) is the
measurement matrix with unit-norm columns, and $\mathbf{w}$ is an
arbitrary and unknown vector of errors. To define block-sparsity,
as in \cite{EldarKuppinger10}, we model $\mathbf{x}$ as a
concatenation of equal-length blocks
\begin{align}
\mathbf{x}=[\mathbf{x}_1^T\phantom{0}\mathbf{x}_2^T\phantom{0}\ldots\phantom{0}\mathbf{x}_L^T]^T
\end{align}
where $\mathbf{x}_l\triangleq[x_{(l-1)
d+1}\phantom{0}\ldots\phantom{0} x_{ld}]^T$ is a $d$-dimensional
vector. Clearly, the vector $\mathbf{x}$ has a dimension $n=Ld$,
and the vector is called block $K$-sparse if its block component
$\mathbf{x}_l$ has nonzero Euclidean norm for at most $K$ indices
$l$. Similarly, the measurement matrix $\mathbf{A}$ can be
expressed as a concatenation of column-block matrices
$\{\mathbf{A}_l\}_{l=1}^L$
\begin{align}
\mathbf{A}=[\mathbf{A}_1\phantom{0}\mathbf{A}_2\phantom{0}\ldots\mathbf{A}_L]
\end{align}
where $\mathbf{A}_l\in\mathbb{R}^{m\times d}$. Also, we assume
that the number of rows of $\mathbf{A}$ is an integer multiples of
$d$, i.e. $m=Rd$ with $R$ an integer. The conventional coherence
metric of the measurement matrix $\mathbf{A}$ is defined as
\begin{align}
\mu\triangleq\max_{i\neq j}|\mathbf{a}_i^T\mathbf{a}_j|
\end{align}
where $\mathbf{a}_i$ denotes the $i\text{th}$ column of
$\mathbf{A}$. This coherence metric, albeit useful, is not
sufficient to characterize the block-structure of the sparse
signal. To exploit the block-sparsity property, we define the
block-coherence $\mu_{\text{B}}$ and sub-coherence $\nu$ (these
two concepts were firstly introduced in \cite{EldarKuppinger10}):
\begin{align}
\mu_{\text{B}}\triangleq&\max_{i,j\neq
i}\quad\frac{1}{d}\rho(\mathbf{A}_i^T\mathbf{A}_j) \nonumber\\
\nu\triangleq&\max_{l}\max_{i,j\neq i}\quad
|\mathbf{a}_i^T\mathbf{a}_j|, \qquad
\mathbf{a}_i,\mathbf{a}_j\in\mathbf{A}_l
\end{align}
where $\rho(\mathbf{X})$ denotes the spectral norm of
$\mathbf{X}$, which is defined as the square root of the maximum
eigenvalue of $\mathbf{X}^T\mathbf{X}$, i.e.
$\sqrt{\lambda_{\text{max}}(\mathbf{X}^T\mathbf{X})}$. Related
properties of the block-coherence $\mu_{\text{B}}$ can be found in
\cite{EldarKuppinger10}. We see that $\mu_{\text{B}}$ quantifies
the coherence between blocks of $\mathbf{A}$, while the coherence
within blocks is characterized by the sub-coherence $\nu$.

The objective of this paper is to identify sufficient conditions
on the measurement matrix $\mathbf{A}$ (in terms of the
block-coherence $\mu_{\text{B}}$ and the sub-coherence $\nu$), as
well as the signal vector $\mathbf{x}$ and the error vector
$\mathbf{w}$, under which the block-sparsity pattern can be
recovered from the noisy measurements. We are particularly
interested in analyzing the recovery ability of a block-version of
the orthogonal matching pursuit (OMP). OMP is a simple greedy
approximation algorithm developed in
\cite{ChenBillings89,PatiRezaiifar93}. Despite its simplicity, OMP
is a provably good approximation algorithm which achieves
performance close to Lasso in certain scenarios
\cite{Tropp04,TroppGilbert07}. In the following, we briefly
summarize the block-version of the OMP, which is also termed as
block-OMP (BOMP). This BOMP is a slight variant of the original
BOMP that was introduced in \cite{EldarKuppinger10} for noiseless
scenarios.

\emph{BOMP Algorithm:}
\begin{enumerate}
\item Initialize the residual $\mathbf{r}_0=\mathbf{y}$, the index
set $S_0=\varnothing$.
\item At the $t\text{th}$ step ($t\geq 1$), we choose the block
that is best matched to $\mathbf{r}_{t-1}$ according to
\begin{align}
i_t=\arg\max_{i}\|\mathbf{A}_i^T\mathbf{r}_{t-1}\|_2
\label{matching-criterion}
\end{align}
\item Augment the index set and the matrix of chosen blocks:
$S_{t}=S_{t-1}\cup\{i_t\}$ and
$\boldsymbol{\Psi}^{(t)}=[\boldsymbol{\Psi}^{(t-1)}\phantom{0}\mathbf{A}_{i_t}]$.
We use the convention that $\boldsymbol{\Psi}^{(0)}$ is an empty
matrix.
\item Solve a least squares problem to obtain a new signal
estimate
$\mathbf{x}_t=\arg\min_{\mathbf{x}}\|\mathbf{y}-\boldsymbol{\Psi}^{(t)}\mathbf{x}\|_2$
\item Calculate the new residual as
$\mathbf{r}_{t}=\mathbf{y}-\boldsymbol{\Psi}^{(t)}\mathbf{x}_t=\mathbf{y}-\mathcal{P}_{\boldsymbol{\Psi}^{(t)}}\mathbf{y}$,
where
$\mathcal{P}_{\boldsymbol{\Psi}^{(t)}}=\boldsymbol{\Psi}^{(t)}(\boldsymbol{\Psi}^{(t)})^{\dag}$
is the orthogonal projection onto the column space of
$\boldsymbol{\Psi}^{(t)}$, and $^{\dag}$ stands for the
pseudo-inverse.
\item If $\|\mathbf{r}_{t}\|_2\geq\epsilon$, return to Step 2;
otherwise stop.
\end{enumerate}

\section{Block-Sparsity Pattern Recovery Analysis}
Let $\mathbf{x}_{\text{nz}}$ denote a $Kd$ dimensional column
vector constructed by stacking the nonzero block components
$\mathbf{x}_l,\forall \{l|\mathbf{x}_l\neq\mathbf{0}\}$,
$\mathbf{A}_{\text{nz}}\in\mathbb{R}^{m\times Kd}$ denote a
submatrix of $\mathbf{A}$ constructed by concatenating the
column-blocks $\mathbf{A}_l,\forall
\{l|\mathbf{x}_l\neq\mathbf{0}\}$, i.e. the blocks corresponding
to the nonzero $\mathbf{x}_l$, and let
$\mathbf{A}_{\text{z}}\in\mathbb{R}^{m\times (L-K)d}$ stand for a
submatrix of $\mathbf{A}$ constructed by concatenating the
column-blocks $\mathbf{A}_l$ corresponding to zero $\mathbf{x}_l$.
For notational convenience, let $I_1=\{l_1,l_2,\ldots,l_K\}$
denote a set of indices for which
$\mathbf{x}_{l_i}\neq\mathbf{0}$, and
$I_2=\{l_{K+1},l_{K+2},\ldots,l_{L}\}$ denote a set of indices for
which $\mathbf{x}_{l_i}=\mathbf{0}$. Therefore we can write
\begin{align}
\mathbf{x}_{\text{nz}}\triangleq&\left[
\begin{array}{cccc}\mathbf{x}_{l_1}^T&\mathbf{x}_{l_2}^T&\ldots&\mathbf{x}_{l_K}^T\end{array}\right]^{T}
\nonumber\\
\mathbf{A}_{\text{nz}}\triangleq&\left[
\begin{array}{cccc}\mathbf{A}_{l_1}&\mathbf{A}_{l_2}&\ldots&\mathbf{A}_{l_K}\end{array}\right]
\nonumber\\
\mathbf{A}_{\text{z}}\triangleq&\left[
\begin{array}{cccc}\mathbf{A}_{l_{K+1}}&\mathbf{A}_{l_{K+2}}&\ldots&\mathbf{A}_{l_L}\end{array}\right]
\nonumber
\end{align}

The measurements can therefore be written as
\begin{align}
\mathbf{y}=\mathbf{A}_{\text{nz}}\mathbf{x}_{\text{nz}}+\mathbf{w}
\label{model-1}
\end{align}
We can decompose the error vector $\mathbf{w}$ into
$\mathbf{w}=\mathcal{P}_{\mathbf{A}_{\text{nz}}}\mathbf{w}+\mathcal{P}_{\mathbf{A}_{\text{nz}}}^{\perp}\mathbf{w}$,
where
$\mathcal{P}_{\mathbf{A}_{\text{nz}}}=\mathbf{A}_{\text{nz}}\mathbf{A}_{\text{nz}}^{\dagger}$
denotes the orthogonal projection onto the subspace spanned by the
columns of $\mathbf{A}_{\text{nz}}$, and
$\mathcal{P}_{\mathbf{A}_{\text{nz}}}^{\perp}=\mathbf{I}-\mathcal{P}_{\mathbf{A}_{\text{nz}}}$
is the orthogonal projection onto the null space of
$\mathbf{A}_{\text{nz}}^T$. We can further write
\begin{align}
\mathbf{y}=\mathbf{A}_{\text{nz}}\mathbf{x}_{\text{nz}}+\mathbf{w}=&\mathbf{A}_{\text{nz}}\mathbf{x}_{\text{nz}}+
\mathcal{P}_{\mathbf{A}_{\text{nz}}}\mathbf{w}+\mathcal{P}_{\mathbf{A}_{\text{nz}}}^{\perp}\mathbf{w}\nonumber\\
=&\mathbf{A}_{\text{nz}}(\mathbf{x}_{\text{nz}}+
\mathbf{A}_{\text{nz}}^{\dagger}\mathbf{w})+\mathcal{P}_{\mathbf{A}_{\text{nz}}}^{\perp}\mathbf{w}\nonumber\\
\triangleq&\mathbf{A}_{\text{nz}}\mathbf{\tilde{x}}_{\text{nz}}+\mathbf{\tilde{w}}
\label{model}
\end{align}
where
$\mathbf{\tilde{x}}_{\text{nz}}\triangleq\mathbf{x}_{\text{nz}}+
\mathbf{A}_{\text{nz}}^{\dagger}\mathbf{w}$, and
$\mathbf{\tilde{w}}\triangleq\mathcal{P}_{\mathbf{A}_{\text{nz}}}^{\perp}\mathbf{w}$.
Equation (\ref{model}) decomposes the measurements into two
mutually orthogonal components: a signal component
$\mathbf{A}_{\text{nz}}\mathbf{\tilde{x}}_{\text{nz}}$ and a noise
component $\mathbf{\tilde{w}}$. The reason for doing so is that
even the exact signal support (block-sparsity pattern) is known,
there is no way to separate the noise projection term
$\mathbf{A}_{\text{nz}}^{\dagger}\mathbf{w}$ from the true signal
$\mathbf{x}_{\text{nz}}$. Hence it is more convenient to carry out
our analysis based on (\ref{model}) instead of (\ref{model-1}).

Recall that, at each iteration, the BOMP algorithm searches for a
block that is best matched to the residual vector according to
(\ref{matching-criterion}). We can define a greedy selection ratio
that determines whether or not a correct block is selected at each
iteration
\begin{align}
\gamma_t=\frac{\max_{l\in
I_2}\|\mathbf{A}_{l}^T\mathbf{r}_{t-1}\|_2} {\max_{l\in
I_1}\|\mathbf{A}_{l}^T\mathbf{r}_{t-1}\|_2} \label{gsr}
\end{align}
where $\mathbf{r}_{t-1}$ is the residual vector at iteration
$t-1$. Clearly, at each iteration, the algorithm picks an index
whose corresponding block is in $\mathbf{A}_{\text{nz}}$ if
$\gamma_t<1$, otherwise an incorrect index whose corresponding
block is in $\mathbf{A}_{\text{z}}$ is chosen. Since the residual
is orthogonal to the subspace spanned by all the previously chosen
block-columns, no index will be chosen twice. Therefore, in order
to recover the block-sparsity pattern, we need to guarantee
$\gamma_t<1$ throughout the first $K$ iterations, i.e.
$\gamma_t<1,\forall t\leq K$. Here for simplicity, we assume that
the number of nonzero blocks, $K$, is known \emph{a priori}. In
practice, $K$ can be automatically determined by the BOMP
algorithm given the error tolerance $\epsilon$ ($\epsilon$ can be
estimated from the observation noise power in practice). As long
as $K$ is not overestimated, i.e. $\hat{K}\leq K$, we can ensure
that all the chosen indices are from the set of correct indices
$I_1$.

In the following, we derive sufficient conditions that guarantee
$\gamma_t<1$ throughout the first $K$ iterations. Before
proceeding, we define a general mixed $\ell_2/\ell_p$-norm
($p=1,2,\infty$) that will be used throughout this paper. For a
vector
$\mathbf{z}=[\mathbf{z}_1^T\phantom{0}\mathbf{z}_2^T\phantom{0}\ldots\phantom{0}\mathbf{z}_Q^T]^T$
consisting of equal-length blocks with block size $d$, the general
mixed $\ell_2/\ell_p$-norm (with block size $d$) is defined as
\begin{align}
\|\mathbf{z}\|_{2,p}=\|\mathbf{v}\|_p \qquad \text{where
$v_q=\|\mathbf{z}_q\|_2$}
\end{align}
Correspondingly, for a matrix $\mathbf{X}\in\mathbb{R}^{Ud\times
Qd}$, where $U$ and $Q$ can be any positive integers, the mixed
matrix norm (with block size $d$) is defined as
\begin{align}
\|\mathbf{X}\|_{2,p}=\max_{\mathbf{z}\neq\mathbf{0}}\frac{\|\mathbf{X}\mathbf{z}\|_{2,p}}{\|\mathbf{z}\|_{2,p}}
\end{align}
Resorting to this general mixed $\ell_2/\ell_p$-norm (with block
size $d$) definition, the greedy selection ratio defined in
(\ref{gsr}) can be re-expressed as
\begin{align}
\gamma_t=\frac{\max_{\{l:\mathbf{x}_l=\mathbf{0}\}}\|\mathbf{A}_{l}^T\mathbf{r}_{t-1}\|_2}
{\max_{\{l:\mathbf{x}_l\neq\mathbf{0}\}}\|\mathbf{A}_{l}^T\mathbf{r}_{t-1}\|_2}
=\frac{\|\mathbf{A}_{\text{z}}^T\mathbf{r}_{t-1}\|_{2,\infty}}{\|\mathbf{A}_{\text{nz}}^T\mathbf{r}_{t-1}\|_{2,\infty}}
\end{align}

Suppose that the BOMP algorithm has successfully executed the
first $k$ ($k<K$) iterations with residual
\begin{align}
\mathbf{r}_k=\mathbf{y}-\mathcal{P}_{\boldsymbol{\Phi}_1}\mathbf{y}
\label{residual}
\end{align}
where $\boldsymbol{\Phi}_1\in\mathbb{R}^{m\times kd}$ is a matrix
constructed by concatenating the $k$ block-columns chosen from the
previous $k$ iterations, and
$\mathcal{P}_{\boldsymbol{\Phi}_1}=\boldsymbol{\Phi}_1\boldsymbol{\Phi}_1^{\dag}$
is the orthogonal projection onto the column space of
$\boldsymbol{\Phi}_1$. Note that $\boldsymbol{\Phi}_1$ is a
sub-matrix of $\mathbf{A}_{\text{nz}}$ since we assume that the
algorithm selected the correct indices during the first $k$
iterations. Let $\boldsymbol{\Phi}_2$ be a matrix constructed by
concatenating the remaining $K-k$ column-blocks in
$\mathbf{A}_{\text{nz}}$. Without loss of generality, we can write
$\mathbf{A}_{\text{nz}}=[\boldsymbol{\Phi}_1\phantom{0}\boldsymbol{\Phi}_2]$,
i.e.
$\boldsymbol{\Phi}_1\triangleq[\mathbf{A}_{l_1}\phantom{0}\ldots\phantom{0}\mathbf{A}_{l_k}]$,
and
$\boldsymbol{\Phi}_2\triangleq[\mathbf{A}_{l_{k+1}}\phantom{0}\ldots\phantom{0}\mathbf{A}_{l_K}]$.
Also, we write
$\mathbf{\tilde{x}}_{\text{nz}}=[\mathbf{\tilde{x}}_{l_1}^T
\phantom{0}\mathbf{\tilde{x}}_{l_1}^T\phantom{0}\ldots\phantom{0}\mathbf{\tilde{x}}_{l_K}^T]^T=
[\boldsymbol{\phi}_1^T\phantom{0}\boldsymbol{\phi}_2^T ]^T$, where
$\boldsymbol{\phi}_1\triangleq[\mathbf{\tilde{x}}_{l_1}^T\phantom{0}\ldots\phantom{0}\mathbf{\tilde{x}}_{l_k}^T]^T$,
and
$\boldsymbol{\phi}_2\triangleq[\mathbf{\tilde{x}}_{l_{k+1}}^T\phantom{0}\ldots\phantom{0}\mathbf{\tilde{x}}_{l_K}^T]^T$.
Substituting (\ref{model}) into (\ref{residual}), the residual can
be written as
\begin{align}
\mathbf{r}_k=&\mathbf{A}_{\text{nz}}\mathbf{\tilde{x}}_{\text{nz}}+\mathbf{\tilde{w}}-\mathcal{P}_{\boldsymbol{\Phi}_1}(
\mathbf{A}_{\text{nz}}\mathbf{\tilde{x}}_{\text{nz}}+\mathbf{\tilde{w}})\nonumber\\
\stackrel{(a)}{=}&\mathbf{A}_{\text{nz}}\mathbf{\tilde{x}}_{\text{nz}}-\mathcal{P}_{\boldsymbol{\Phi}_1}
\mathbf{A}_{\text{nz}}\mathbf{\tilde{x}}_{\text{nz}}+\mathbf{\tilde{w}} \nonumber\\
\stackrel{(b)}{=}&\boldsymbol{\Phi}_2\boldsymbol{\phi}_2-\mathcal{P}_{\boldsymbol{\Phi}_1}
\boldsymbol{\Phi}_2\boldsymbol{\phi}_2+\mathbf{\tilde{w}} \nonumber\\
\stackrel{(c)}{=}&\mathbf{\tilde{r}}_k+\mathbf{\tilde{w}}
\label{residual-p}
\end{align}
where $(a)$ comes from the fact that $\mathbf{\tilde{w}}$ is
orthogonal to the column space of $\boldsymbol{\Phi}_1$, and $(b)$
comes by noting that
$\mathcal{P}_{\boldsymbol{\Phi}_1}\boldsymbol{\Phi}_1=\boldsymbol{\Phi}_1$,
and in $(c)$ we define
$\mathbf{\tilde{r}}_k\triangleq\boldsymbol{\Phi}_2\boldsymbol{\phi}_2-\mathcal{P}_{\boldsymbol{\Phi}_1}
\boldsymbol{\Phi}_2\boldsymbol{\phi}_2$. Using this result, the
greedy selection ratio at iteration $k+1$ becomes
\begin{align}
\gamma_{k+1}=&\frac{\|\mathbf{A}_{\text{z}}^T\mathbf{r}_k\|_{2,\infty}}{\|\mathbf{A}_{\text{nz}}^T\mathbf{r}_k\|_{2,\infty}}
=\frac{\|\mathbf{A}_{\text{z}}^T(\mathbf{\tilde{r}}_k+\mathbf{\tilde{w}})\|_{2,\infty}}
{\|\mathbf{A}_{\text{nz}}^T(\mathbf{\tilde{r}}_k+\mathbf{\tilde{w}})\|_{2,\infty}}
\nonumber\\
=&\frac{\|\mathbf{A}_{\text{z}}^T(\mathbf{\tilde{r}}_k+\mathbf{\tilde{w}})\|_{2,\infty}}
{\|\mathbf{A}_{\text{nz}}^T\mathbf{\tilde{r}}_k\|_{2,\infty}}
\stackrel{(a)}{\leq}
\frac{\|\mathbf{A}_{\text{z}}^T\mathbf{\tilde{r}}_k\|_{2,\infty}+\|\mathbf{A}_{\text{z}}^T\mathbf{\tilde{w}}\|_{2,\infty}}
{\|\mathbf{A}_{\text{nz}}^T\mathbf{\tilde{r}}_k\|_{2,\infty}} \nonumber\\
\stackrel{(b)}{=}&\frac{\|\mathbf{A}_{\text{z}}^T\mathcal{P}_{\mathbf{A}_{\text{nz}}}\mathbf{\tilde{r}}_k\|_{2,\infty}}
{\|\mathbf{A}_{\text{nz}}^T\mathbf{\tilde{r}}_k\|_{2,\infty}}+\frac{\|\mathbf{A}_{\text{z}}^T\mathbf{\tilde{w}}\|_{2,\infty}}
{\|\mathbf{A}_{\text{nz}}^T\mathbf{\tilde{r}}_k\|_{2,\infty}}\nonumber\\
=&\frac{\|\mathbf{A}_{\text{z}}^T(\mathbf{A}_{\text{nz}}^{\dag})^T\mathbf{A}_{\text{nz}}^T\mathbf{\tilde{r}}_k\|_{2,\infty}}
{\|\mathbf{A}_{\text{nz}}^T\mathbf{\tilde{r}}_k\|_{2,\infty}}+\frac{\|\mathbf{A}_{\text{z}}^T\mathbf{\tilde{w}}\|_{2,\infty}}
{\|\mathbf{A}_{\text{nz}}^T\mathbf{\tilde{r}}_k\|_{2,\infty}}\nonumber\\
\leq
&\|\mathbf{A}_{\text{z}}^T(\mathbf{A}_{\text{nz}}^{\dag})^T\|_{2,\infty}+
\frac{\|\mathbf{A}_{\text{z}}^T\mathbf{\tilde{w}}\|_{2,\infty}}
{\|\mathbf{A}_{\text{nz}}^T\mathbf{\tilde{r}}_k\|_{2,\infty}}
\label{ratio}
\end{align}
where $(a)$ comes from the fact the general mixed
$\ell_2/\ell_p$-norm satisfies the triangle inequality:
$\|\mathbf{a}+\mathbf{b}\|_{2,\infty}\leq\|\mathbf{a}\|_{2,\infty}+\|\mathbf{b}\|_{2,\infty}$,
which can be readily verified, $(b)$ follows from
$\mathcal{P}_{\mathbf{A}_{\text{nz}}}\mathbf{\tilde{r}}_k=\mathbf{\tilde{r}}_k$
since $\mathbf{\tilde{r}}_k$ lies in the column space of
$\mathbf{A}_{\text{nz}}$. Our objective is to identify conditions
assuring $\gamma_{k+1}<1$.

If the measurement process is perfect and noise-free, that is,
$\mathbf{y}=\mathbf{A}\mathbf{x}$, then the greedy selection ratio
is simply upper bounded by
\begin{align}
\gamma_{k+1}\leq\|\mathbf{A}_{\text{z}}^T(\mathbf{A}_{\text{nz}}^{\dag})^T\|_{2,\infty}
\end{align}
Furthermore, it has been shown in \cite[Lemma 4]{EldarKuppinger10}
that
$\|\mathbf{A}_{\text{z}}^T(\mathbf{A}_{\text{nz}}^{\dag})^T\|_{2,\infty}$
is upper bounded by
\begin{align}
\|\mathbf{A}_{\text{z}}^T(\mathbf{A}_{\text{nz}}^{\dag})^T\|_{2,\infty}
\leq\frac{Kd\mu_{\text{B}}}{1-(d-1)\nu-(K-1)d\mu_{\text{B}}}
\label{eq6}
\end{align}
Therefore the condition $\gamma_{k+1}<1$ holds universally if the
block-coherence $\mu_{\text{B}}$ and sub-coherence $\nu$
associated with the dictionary $\mathbf{A}$ satisfies
\begin{align}
\frac{Kd\mu_{\text{B}}}{1-(d-1)\nu-(K-1)d\mu_{\text{B}}}< 1
\label{eq4}
\end{align}
Since, in practice, measurements are inevitably contaminated with
noise and underlying uncertainties, it is thus important to
understand the effect of measurement noise on the block-sparsity
pattern recovery. Apparently, when noise is present, condition
(\ref{eq4}) alone cannot guarantee the exact recovery of the
block-sparsity pattern. Instead, from (\ref{ratio}), we see that,
to assure $\gamma_{p+1}<1$, we need
\begin{align}
\|\mathbf{A}_{\text{z}}^T(\mathbf{A}_{\text{nz}}^{\dag})^T\|_{2,\infty}+
\frac{\|\mathbf{A}_{\text{z}}^T\mathbf{\tilde{w}}\|_{2,\infty}}
{\|\mathbf{A}_{\text{nz}}^T\mathbf{\tilde{r}}_k\|_{2,\infty}}<1
\label{eq5}
\end{align}
The inequality (\ref{eq5}) has to hold valid for $0\leq k\leq K-1$
in order to ensure that the BOMP algorithm chooses the correct
indices throughout the first $K$ iterations. In the following, we
provide sufficient conditions which guarantee (\ref{eq5}) for
$0\leq k\leq K-1$. The results are summarized as follows.

\newtheorem{theorem}{Theorem}
\begin{theorem} \label{theorem1}
Let
\begin{align}
\omega\triangleq\|\mathbf{A}^T\mathbf{\tilde{w}}\|_{2,\infty}=\max_l\|\mathbf{A}_l^T\mathbf{\tilde{w}}\|_2
\end{align}
denote the maximum correlation between the column block
$\mathbf{A}_l$ and the residual noise component
$\mathbf{\tilde{w}}$. Let
\begin{align}
x_{\text{b,min}}\triangleq\min_{l\in
I_1}\|\mathbf{\tilde{x}}_l\|_2
\end{align}
the minimum $\ell_2$-norm of the non-zero signal block components.
Suppose that the following conditions are satisfied
\begin{align}
\text{(i)} &\quad 1-(d-1)\nu-(2K-1)d\mu_{\text{B}}>0 \nonumber\\
\text{(ii)} &\quad
\frac{[1-(d-1)\nu-(2K-1)d\mu_{\text{B}}]^2}{1-(d-1)\nu-(K-1)d\mu_{\text{B}}}>\frac{\omega}{x_{\text{b,min}}}
\label{theorem1:condition}
\end{align}
then we can guarantee that the BOMP algorithm selects indices from
$I_1$ throughout the first $K$ iterations. If the error tolerance
$\epsilon$ is chosen such that the algorithm stops at the end of
iteration $K$, then the BOMP recovers the exact block-sparsity
pattern.
\end{theorem}
\begin{proof}
See Appendix \ref{appA}.
\end{proof}

Theorem \ref{theorem1} is a generalization of the results
presented in \cite{EldarKuppinger10} which considered block-sparse
signal recovery from noise-free measurements. To see this, for the
noiseless case, we have $\omega=0$, and hence the condition
(\ref{theorem1:condition}) is simplified as
\begin{align}
1-(d-1)\nu-(2K-1)d\mu_{\text{B}}>0
\end{align}
which is exactly the recovery condition provided in
\cite{EldarKuppinger10} for block-sparse signal recovery. On the
other hand, for the noisy case, the success of the BOMP algorithm
not only depends on the block-coherence $\mu_{\text{B}}$ and the
sub-coherence $\nu$, but also depends on the ratio of the maximum
correlation (between the column block $\mathbf{A}_l$ and the
residual noise component $\mathbf{\tilde{w}}$) to the minimum
$\ell_2$-norm of the nonzero signal block components
$\mathbf{\tilde{x}}_{l},\forall l\in I_1$. The importance of the
minimum nonzero signal component in sparsity pattern recovery has
been highlighted in \cite{Wainwright09,FletcherRangan09}. In
particular, \cite{Wainwright09} showed that both the sufficient
and necessary conditions require control of the minimum nonzero
signal component. Our result suggests that, for block-sparse
signal recovery, the minimum $\ell_2$-norm of the nonzero signal
block components, instead of the minimum magnitude of an entry, is
the key quantity that controls the block subset selection.

Also, we observe that the left-hand side of the second condition
in (\ref{theorem1:condition}) is strictly less than one. Therefore
the ratio $\omega/x_{\text{b,min}}$ cannot be greater than one,
otherwise the condition cannot be met, irrespective of the choice
of the sub-coherence $\nu$ and the block-coherence
$\mu_{\text{B}}$. The deterministic condition
(\ref{theorem1:condition}), however, guarantees recovery of the
sparsity pattern under the worst-case scenario and therefore is
very pessimistic. If we take a probabilistic analysis (as in
\cite{FletcherRangan09a}) that ensures a probabilistic recovery,
the condition can be significantly relaxed. This could be a
direction of our future study.

\section{Discussions}
We note that in this paper, as in \cite{EldarKuppinger10},
block-sparsity is explicitly exploited to yield a more relaxed
condition imposed on the measurement matrix, and therefore lead to
a guaranteed recovery for a potentially higher sparsity level. If
the block-sparse signal is treated as a conventional $Kd$-sparse
vector without exploiting knowledge of the block-sparsity
structure, sufficient conditions for exact sparsity pattern
recovery using OMP are given in \cite[Theorem 18]{Tropp06} and can
be formulated as (by combining the first and the third equation in
\cite[Theorem 18]{Tropp06})
\begin{align}
\text{(i)}&\quad 1-2Kd\mu>0 \nonumber\\
\text{(ii)}&\quad
\frac{(1-2Kd\mu)^2}{1-Kd\mu}>\frac{\|\mathbf{A}^T\mathbf{\tilde{w}}\|_{\infty}}{x_{\text{min}}}
\label{recovery-condition-OMP}
\end{align}
where $x_{\text{min}}$ denotes the minimum magnitude of the
nonzero signal elements in $\mathbf{\tilde{x}}_{\text{nz}}$. When
$d=1$, block-sparsity reduces to conventional sparsity and we have
$\nu=0$, $\mu_{\text{B}}=\mu$. The condition
(\ref{theorem1:condition}) is simplified as
\begin{align}
\text{(i)}&\quad 1-(2K-1)d\mu>0 \nonumber\\
\text{(ii)}&\quad
\frac{(1-(2K-1)d\mu)^2}{1-(K-1)d\mu}>\frac{\|\mathbf{A}^T\mathbf{\tilde{w}}\|_{\infty}}{x_{\text{min}}}
\label{eq8}
\end{align}
which is the same as (\ref{recovery-condition-OMP}) except that
$2K$ and $K$ in the numerator and denominator are replaced by
$2K-1$ and $K-1$, respectively (It can be easily verified that
(\ref{eq8}) is slightly loose than
(\ref{recovery-condition-OMP})). When $d>1$, in the special case
that the columns of $\mathbf{A}_l$ are orthonormal for each $l$,
we have $\nu=0$ and therefore the recovery condition
(\ref{theorem1:condition}) becomes
\begin{align}
\text{(i)} &\quad 1-(2K-1)d\mu_{\text{B}}>0 \nonumber\\
\text{(ii)} &\quad
\frac{[1-(2K-1)d\mu_{\text{B}}]^2}{1-(K-1)d\mu_{\text{B}}}>\frac{\omega}{s_{\text{min}}}
\label{recovery-condition-BOMP}
\end{align}
This recovery condition, (\ref{recovery-condition-BOMP}), is less
restrictive than (\ref{recovery-condition-OMP}) since we have
\begin{align}
\frac{[1-(2K-1)d\mu_{\text{B}}]^2}{1-(K-1)d\mu_{\text{B}}}>&
\frac{(1-2Kd\mu_{\text{B}})^2}{1-Kd\mu_{\text{B}}}
\stackrel{(a)}{\geq}\frac{(1-2Kd\mu)^2}{1-Kd\mu}\nonumber\\
>&\frac{\|\mathbf{A}^T\mathbf{\tilde{w}}\|_{\infty}}{x_{\text{min}}}\stackrel{(b)}{\geq}
\frac{\omega}{x_{\text{b,min}}} \label{eq7}
\end{align}
where $(a)$ comes from the fact that $1-2Kd\mu>0$ and
$\mu_{\text{B}}\leq\mu$ \cite[Proposition 2]{EldarKuppinger10},
$(b)$ follows from
$\omega\leq\sqrt{d}\|\mathbf{A}^T\mathbf{\tilde{w}}\|_{\infty}$
and $x_{\text{b,min}}\geq\sqrt{d}x_{\text{min}}$. We see that
through exploiting the block-sparsity, the sparsity pattern
recovery condition is relaxed and we can guarantee a recovery of
sparsity pattern with a higher sparsity level. A close examination
of (\ref{eq7}) reveals that this improvement comes from two
aspects. First, the measurement matrix requires a less restrictive
mutual coherence condition since $\mu_{\text{B}}\leq\mu$. Second,
for the same signal, noise, and measurement matrix, the quantity
$\omega/x_{\text{b,min}}$ is always smaller than or equal to
$\|\mathbf{A}^T\mathbf{\tilde{w}}\|_{\infty}/x_{\text{min}}$,
meaning that exploiting block-sparsity can improve the ability of
detecting weak signals buried in noise.

If the individual blocks $\mathbf{A}_l$ are, however, not
orthonormal, then $\nu>0$, and $\nu$ has to be small in order to
result in a performance gain for block-sparsity recovery as
compared with the conventional sparse recovery. We can also follow
the orthogonalization approach \cite{EldarKuppinger10} to analyze
the general non-orthonormal case. We orthogonalize the individual
blocks $\mathbf{A}_l=\mathbf{\tilde{A}}_l\mathbf{V}_l$, in which
$\mathbf{\tilde{A}}_l$ consists of orthonormal columns, and
$\mathbf{V}_l$ is an invertible matrix. The original dictionary
can therefore be written as
$\mathbf{A}=\mathbf{\tilde{A}}\mathbf{V}$, where $\mathbf{V}$ is a
block-diagonal matrix with blocks $\mathbf{V}_l$. Clearly,
orthogonalization preserves the block-sparsity level. The
comparison that is meaningful here is between the recovery based
on the original model without exploiting block-sparsity and the
recovery based on the orthogonalized model taking block-sparsity
into account. For the orthogonalized dictionary
$\mathbf{\tilde{A}}$, we have $\nu(\mathbf{\tilde{A}})=0$.
Therefore we are only concerned about the relation between $\mu$
before orthogonalization and $\mu_{\text{B}}$ after
orthogonalization, which are denoted by $\mu(\mathbf{A})$ and
$\mu_{\text{B}}(\mathbf{\tilde{A}})$ respectively. Although an
exact relation between $\mu(\mathbf{A})$ and
$\mu_{\text{B}}(\mathbf{\tilde{A}})$ is difficult to derive, it
has been shown in \cite{EldarKuppinger10} that if $d>RL/(L-R)$,
then we have
$\mu(\mathbf{A})\geq\mu_{\text{B}}(\mathbf{\tilde{A}})$. Hence
even for general dictionaries, exploiting block-sparsity still
leads to a guaranteed sparsity pattern recovery for a potentially
higher sparsity level by properly choosing the number of
measurements to satisfy $d>RL/(L-R)$.

We explore the connection and difference between our work and
\cite{TroppGilbert06,GribonvalRauhut04}. In
\cite{TroppGilbert06,GribonvalRauhut04}, the problem of
simultaneous sparse approximation has been extensively studied and
many interesting and elegant results were obtained under different
performance metrics. Among them, the result most related to our
work is \cite[Theorem 5.3]{TroppGilbert06}, which presents a
sufficient condition for simultaneous sparse pattern recovery. The
difference between our work and
\cite{TroppGilbert06,GribonvalRauhut04} lies in two aspects.
First, the problem considered in this paper is more general than
that of \cite{TroppGilbert06,GribonvalRauhut04} since simultaneous
sparse approximation is a special form of block-sparse signal
recovery with the measurement matrix having a block-diagonal
structure and identical diagonal blocks. Second, block-sparsity is
exploited in our paper to improve the recovery ability of dealing
with a higher sparsity level, whereas for
\cite{TroppGilbert06,GribonvalRauhut04}, the simultaneous sparse
approximation does not lead to a more relaxed condition on the
dictionary as compared with the conventional single vector sparse
approximation.

\section{Numerical Results}
We present numerical results to illustrate the sparsity pattern
recovery performance of the BOMP algorithm. In the simulations,
the dictionary is randomly generated with each entry independently
drawn from Gaussian distribution with zero mean and unit variance.
We then normalize each column of the dictionary to satisfy the
unit-norm constraint. The dictionary is divided into consecutive
blocks of length $d$. The support set of the block-sparse signal
is randomly chosen according to a uniform distribution, and the
signals on the support set are i.i.d. Gaussian random variables
with zero mean and unit variance. The measurement noise vector is
randomly generated with each entry drawn from Gaussian
distribution with zero mean and variance $\sigma_w^2$.

To show the effectiveness of the BOMP algorithm, we compare it
with the OMP algorithm that does not take block-sparsity into
account. Fig. \ref{fig1} shows the sparsity pattern recovery
success rate as a function of the block-sparsity level, $K$. The
sparsity pattern recovery is considered successful only if the
algorithm determines all the correct support indices in the first
$K$ steps for the BOMP or in the first $Kd$ steps for the OMP,
supposing the block-sparsity level, $K$, is known \emph{a priori}.
The results are averaged over $1000$ Monte Carlo runs, with the
dictionary, the signal, and the noise randomly generated for each
run. From Fig. \ref{fig1}, we observe that for both the BOMP and
the OMP algorithms, the success rate decreases as the
block-sparsity level, $K$, increases. Also, it can be seen that
the BOMP algorithm presents a significant performance improvement
over the OMP. The result corroborate our theoretical claim that
exploiting block-sparsity can lead to an improved recovery
ability. Fig. \ref{fig3} depicts the success rate of the BOMP
algorithm under different noise power levels. We see that as the
noise power increases, the recovery performance degrades. This
observation is quite intuitive and coincides with our theoretical
result since a higher noise power calls for a stricter requirement
on the measurement matrix in order to satisfy the condition
(\ref{theorem1:condition}).

\begin{figure}[!t]
\centering
\includegraphics[width=9cm]{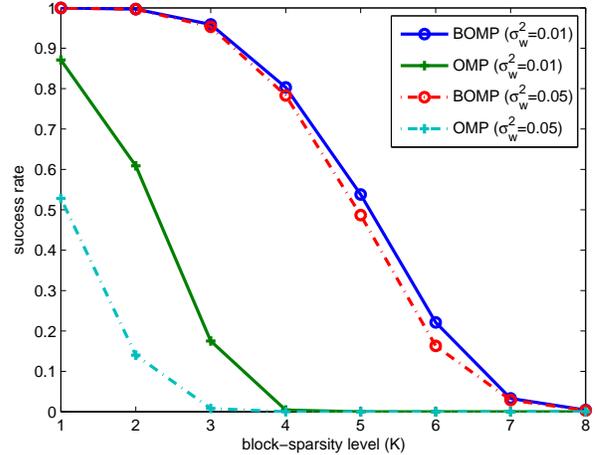}
\caption{Sparsity pattern recovery success rates of OMP and BOMP
algorithms vs. block sparsity level, $m=40$, $n=400$, $d=4$, and
$L=100$.} \label{fig1}
\end{figure}

\begin{figure}[!t]
\centering
\includegraphics[width=9cm]{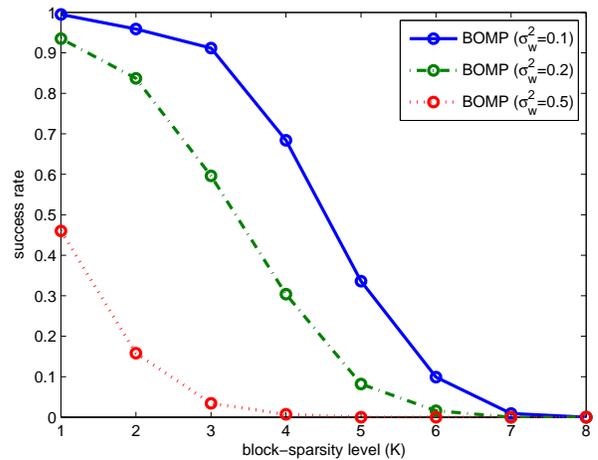}
\caption{Sparsity pattern recovery success rate of BOMP algorithm
vs. block sparsity level, $m=40$, $n=400$, $d=4$, and $L=100$.}
\label{fig3}
\end{figure}

\section{Conclusion}
We studied the problem of recovering the sparsity pattern of
block-sparse signals from noise-corrupted measurements. Our
results showed that even in the presence of noise, the
block-sparsity pattern can still be completely recovered via a
block-version of the OMP algorithm when certain conditions are
satisfied. Also, our analysis revealed that exploiting
block-sparsity can lead to a guaranteed recovery for a potentially
higher sparsity level. This theoretical claim was also
corroborated by our numerical results.

\useRomanappendicesfalse
\appendices

\section{Proof of Theorem \ref{theorem1}} \label{appA}
To prove Theorem \ref{theorem1}, we only need to prove that
(\ref{eq5}) holds for $0\leq k\leq K-1$ given the condition
(\ref{theorem1:condition}) satisfied. To this goal, we first
derive an upper bound on the second term on the left-hand side
(L.H.S.) of (\ref{eq5}).

The numerator of the second term on the L.H.S. of (\ref{eq5}) is
upper bounded by
\begin{align}
\|\mathbf{A}_{\text{z}}^T\mathbf{\tilde{w}}\|_{2,\infty}\leq
\|\mathbf{A}^T\mathbf{\tilde{w}}\|_{2,\infty}=\omega
\end{align}
To derive an upper bound on the second term on the L.H.S. of
(\ref{eq5}), we need to obtain a lower bound on its denominator in
terms of the block coherence parameter $\mu_{\text{B}}$ and the
sub-coherence parameter $\nu$. We have
\begin{align}
\|\mathbf{A}_{\text{nz}}^T\mathbf{\tilde{r}}_k\|_{2,\infty}=&
\|\mathbf{A}_{\text{nz}}^T(\boldsymbol{\Phi}_2\boldsymbol{\phi}_2-\mathcal{P}_{\boldsymbol{\Phi}_1}
\boldsymbol{\Phi}_2\boldsymbol{\phi}_2)\|_{2,\infty} \nonumber\\
\stackrel{(a)}{\geq}&
\|\mathbf{A}_{\text{nz}}^T\boldsymbol{\Phi}_2\boldsymbol{\phi}_2\|_{2,\infty}-
\|\mathbf{A}_{\text{nz}}^T\mathcal{P}_{\boldsymbol{\Phi}_1}
\boldsymbol{\Phi}_2\boldsymbol{\phi}_2\|_{2,\infty}
\label{appA:eq1}
\end{align}
where $(a)$ comes from the general mixed $\ell_2/\ell_p$-norm
triangle inequality. The first term on the right-hand side
(R.H.S.) of (\ref{appA:eq1}) can be further lower bounded as
\begin{align}
&\|\mathbf{A}_{\text{nz}}^T\boldsymbol{\Phi}_2\boldsymbol{\phi}_2\|_{2,\infty}\nonumber\\
=&\max_{i\in I_1}\|\mathbf{A}_i^T
\boldsymbol{\Phi}_2\boldsymbol{\phi}_2\|_2 =\max_{i\in
I_1}\bigg\|\sum_{j=k+1}^K\mathbf{A}_i^T\mathbf{A}_{l_j}\mathbf{\tilde{x}}_{l_j}\bigg\|_2
\nonumber\\
\geq &\max_{i\in\{l_{k+1},\ldots,l_K\}}
\bigg\{\|\mathbf{A}_{i}^T\mathbf{A}_{i}\mathbf{\tilde{x}}_{i}\|_2
\nonumber\\
&\qquad \qquad\qquad-\sum_{\{j|l_j\neq i, k+1\leq j\leq
K\}}\|\mathbf{A}_{i}^T\mathbf{A}_{l_j}\mathbf{\tilde{x}}_{l_j}\|_2\bigg\}
\nonumber\\
\stackrel{(a)}{\geq} & \max_{i\in\{l_{k+1},\ldots,l_K\}}
\bigg\{(1-(d-1)\nu)\|\mathbf{\tilde{x}}_{i}\|_2 \nonumber\\
& \qquad \qquad \qquad -d\mu_{\text{B}}\sum_{\{j|l_j\neq i,
k+1\leq j\leq
K\}}\|\mathbf{\tilde{x}}_{l_j}\|_2\bigg\} \nonumber\\
\geq &
(1-(d-1)\nu-(K-k-1)d\mu_{\text{B}})\max_{i\in\{l_{k+1},\ldots,l_K\}}\|\mathbf{\tilde{x}}_{i}\|_2
\label{appA:eq2})
\end{align}
where $(a)$ comes from the fact that
$\lambda_{\text{min}}(\mathbf{A}_{i}^T\mathbf{A}_{i})\geq
1-(d-1)\nu$ (this fact comes directly from the Gershgorin Circle
Theorem), and $\rho(\mathbf{A}_{i}^T\mathbf{A}_{j})\leq
d\mu_{\text{B}}$ for $i\neq j$. On the other hand, the second term
on the R.H.S. of (\ref{appA:eq1}) can be upper bounded by (Please
see Appendix \ref{appB} for the detailed derivation)
\begin{align}
\|\mathbf{A}_{\text{nz}}^T\mathcal{P}_{\boldsymbol{\Phi}_1}
\boldsymbol{\Phi}_2\boldsymbol{\phi}_2\|_{2,\infty} \leq&
d\mu_{\text{B}}(K-k)\max_{i\in\{l_{k+1},\ldots,l_K\}}\|\mathbf{\tilde{x}}_{i}\|_2
\label{appA:eq3}
\end{align}
Combining (\ref{appA:eq1})--(\ref{appA:eq3}), (\ref{appA:eq1}) is
further lower bounded by
\begin{align}
&\|\mathbf{A}_{\text{nz}}^T\mathbf{\tilde{r}}_k\|_{2,\infty}
\nonumber\\
\geq&(1-(d-1)\nu-(2K-2k-1)d\mu_{\text{B}})\max_{i\in\{l_{k+1},\ldots,l_K\}}\|\mathbf{\tilde{x}}_{i}\|_2
\nonumber\\
\geq&(1-(d-1)\nu-(2K-2k-1)d\mu_{\text{B}})x_{\text{b,min}}
\label{appA:eq6}
\end{align}
Since (\ref{eq4}) is a necessary condition for (\ref{eq5}), we
should always have $1-(d-1)\nu-(2K-1)d\mu_{\text{B}}>0$. Therefore
we can guarantee that the above derived lower bound is positive.
Consequently an upper bound on the second term on the L.H.S. of
(\ref{eq5}) can be derived and given as
\begin{align}
\frac{\|\mathbf{A}_{\text{z}}^T\mathbf{\tilde{w}}\|_{2,\infty}}
{\|\mathbf{A}_{\text{nz}}^T\mathbf{\tilde{r}}_k\|_{2,\infty}}\leq&\frac{\omega}
{(1-(d-1)\nu-(2K-2k-1)d\mu_{\text{B}})x_{\text{b,min}}} \nonumber\\
\leq& \frac{\omega}
{(1-(d-1)\nu-(2K-1)d\mu_{\text{B}})x_{\text{b,min}}}\label{appA:eq4}
\end{align}
We see that the first and the second term on the L.H.S. of
(\ref{eq5}) are respectively upper bounded by (\ref{eq6}) and
(\ref{appA:eq4}). Therefore (\ref{eq5}) is guaranteed if the
summation of these two upper bounds are smaller than unity, i.e.
\begin{align}
&\frac{Kd\mu_{\text{B}}}{1-(d-1)\nu-(K-1)d\mu_{\text{B}}}\nonumber\\
&+\frac{\omega}
{(1-(d-1)\nu-(2K-1)d\mu_{\text{B}})x_{\text{b,min}}}<1
\label{appA:eq5}
\end{align}
A further transformation easily shows that (\ref{appA:eq5}) and
(\ref{theorem1:condition}) are equivalent (note that the condition
$1-(d-1)\nu-(2K-1)d\mu_{\text{B}}>0$ has to be explicitly
indicated to assure (\ref{eq4}) and to assure the positiveness of
the lower bound (\ref{appA:eq6})). The proof is completed here.

\section{Derivation of Equation (\ref{appA:eq3})} \label{appB}
Clearly we have
\begin{align}
\|\mathbf{A}_{\text{nz}}^T\mathcal{P}_{\boldsymbol{\Phi}_1}
\boldsymbol{\Phi}_2\boldsymbol{\phi}_2\|_{2,\infty} =\max_{i\in
I_1}\|\mathbf{A}_i^T \mathcal{P}_{\boldsymbol{\Phi}_1}
\boldsymbol{\Phi}_2\boldsymbol{\phi}_2\|_2
\end{align}
We consider two different cases. If $\mathbf{A}_i$ is a
column-block of $\boldsymbol{\Phi}_1$, i.e.
$i\in\{l_1,\ldots,l_k\}$, then for any index $i$, we have
\begin{align}
&\|\mathbf{A}_i^T\mathcal{P}_{\boldsymbol{\Phi}_1}
\boldsymbol{\Phi}_2\boldsymbol{\phi}_2\|_2\stackrel{(a)}{=}\|\mathbf{A}_i^T\boldsymbol{\Phi}_2\boldsymbol{\phi}_2\|_2
=\bigg\|\sum_{j=k+1}^K\mathbf{A}_i^T\mathbf{A}_{l_j}\mathbf{\tilde{x}}_{l_j}
\bigg\|_2 \nonumber\\
\leq&\sum_{j=k+1}^K\|\mathbf{A}_i^T\mathbf{A}_{l_j}\mathbf{\tilde{x}}_{l_j}
\|_2 \leq d\mu_{\text{B}}\sum_{j=k+1}^K\|\mathbf{\tilde{x}}_{l_j}
\|_2 \nonumber\\
\leq&d\mu_{\text{B}}(K-k)\max_{i\in\{l_{k+1},\ldots,l_K\}}\|\mathbf{\tilde{x}}_{i}\|_2
\label{appB:eq1}
\end{align}
where $(a)$ comes from the fact that
$\boldsymbol{\Phi}_1^T\mathcal{P}_{\boldsymbol{\Phi}_1}=\boldsymbol{\Phi}_1^T\boldsymbol{\Phi}_1
(\boldsymbol{\Phi}_1^T\boldsymbol{\Phi}_1)^{-1}\boldsymbol{\Phi}_1^T=\boldsymbol{\Phi}_1^T$,
and therefore $\mathbf{A}_i^T
\mathcal{P}_{\boldsymbol{\Phi}_1}=\mathbf{A}_i^T$ for
$i\in\{l_1,\ldots,l_k\}$. On the other hand, if $\mathbf{A}_i$ is
a column-block of $\boldsymbol{\Phi}_2$, i.e.
$i\in\{l_{k+1},\ldots,l_K\}$. We show that
\begin{align}
\max_{i\in\{l_1,\ldots,l_k\}}\|\mathbf{A}_i^T\mathcal{P}_{\boldsymbol{\Phi}_1}
\boldsymbol{\Phi}_2\boldsymbol{\phi}_2\|_2\geq\max_{i\in\{l_{k+1},\ldots,l_K\}}
\|\mathbf{A}_i^T\mathcal{P}_{\boldsymbol{\Phi}_1}
\boldsymbol{\Phi}_2\boldsymbol{\phi}_2\|_2 \label{appB:eq2}
\end{align}
To this goal, let
$\mathbf{z}\triangleq\boldsymbol{\Phi}_1^{\dag}\boldsymbol{\Phi}_2\boldsymbol{\phi}_2=[\mathbf{z}_1^T
\phantom{0}\ldots\phantom{0}\mathbf{z}_k^T]^T$, the term on the
L.H.S. of (\ref{appB:eq2}) is lower bounded as
\begin{align}
&\max_{i\in\{l_1,\ldots,l_k\}}\|\mathbf{A}_i^T\mathcal{P}_{\boldsymbol{\Phi}_1}
\boldsymbol{\Phi}_2\boldsymbol{\phi}_2\|_2 =
\max_{i\in\{l_1,\ldots,l_k\}}\|\mathbf{A}_i^T\boldsymbol{\Phi}_1\mathbf{z}\|_2
\nonumber\\
=&\max_{i\in\{l_1,\ldots,l_k\}}\bigg\|\sum_{j=1}^k\mathbf{A}_i^T\mathbf{A}_{l_j}\mathbf{z}_j\bigg\|_2
\nonumber\\
\stackrel{(a)}{\geq}&\|\mathbf{A}_{l_q}^T\mathbf{A}_{l_q}\mathbf{z}_{q}\|_2-
\sum_{\{j|j\neq q,1\leq j\leq
k\}}\|\mathbf{A}_{l_q}^T\mathbf{A}_{l_j}\mathbf{z}_j\|_2
\nonumber\\
\geq&(1-(d-1)\nu)\|\mathbf{z}_q\|_2-d\mu_{\text{B}}
\sum_{\{j|j\neq q,1\leq j\leq k\}}\|\mathbf{z}_j\|_2
\nonumber\\
\geq&(1-(d-1)\nu-(k-1)d\mu_{\text{B}})\|\mathbf{z}_q\|_2
\label{appB:eq3}
\end{align}
where in $(a)$, the index $q$ is chosen such that $\mathbf{z}_{q}$
has the maximum $\ell_2$-norm among $\{\mathbf{z}_{i}\}_{i=1}^k$.
The term on the R.H.S. of (\ref{appB:eq2}) is upper bounded by
\begin{align}
&\max_{i\in\{l_{k+1},\ldots,l_K\}}
\|\mathbf{A}_i^T\mathcal{P}_{\boldsymbol{\Phi}_1}
\boldsymbol{\Phi}_2\boldsymbol{\phi}_2\|_2=\max_{i\in\{l_{k+1},\ldots,l_K\}}
\|\mathbf{A}_i^T\boldsymbol{\Phi}_1\mathbf{z}\|_2 \nonumber\\
=&\max_{i\in\{l_{k+1},\ldots,l_K\}}\bigg\|\sum_{j=1}^k\mathbf{A}_i^T\mathbf{A}_{l_j}\mathbf{z}_j\bigg\|_2
\leq\max_{i\in\{l_{k+1},\ldots,l_K\}}\sum_{j=1}^k\|\mathbf{A}_i^T\mathbf{A}_{l_j}\mathbf{z}_j\|_2
\nonumber\\
\leq&d\mu_{\text{B}}\sum_{j=1}^k\|\mathbf{z}_j\|_2 \leq
kd\mu_{\text{B}}\|\mathbf{z}_q\|_2 \label{appB:eq4}
\end{align}
Since we have $1-(d-1)\nu-(2K-1)d\mu_{\text{B}}>0$ in order to
assure the condition (\ref{eq4}) to be satisfied, we can easily
verify that the following always holds for $0\leq k<K$
\begin{align}
(1-(d-1)\nu-(k-1)d\mu_{\text{B}})>kd\mu_{\text{B}}
\label{appB:eq5}
\end{align}
The inequality (\ref{appB:eq2}) comes directly by combining
(\ref{appB:eq3}--\ref{appB:eq5}). Therefore the second term on the
R.H.S. of (\ref{appA:eq1}) is upper bounded by
\begin{align}
\|\mathbf{A}_{\text{nz}}^T\mathcal{P}_{\boldsymbol{\Phi}_1}
\boldsymbol{\Phi}_2\boldsymbol{\phi}_2\|_{2,\infty}=&\max_{i\in
I_1}\|\mathbf{A}_i^T \mathcal{P}_{\boldsymbol{\Phi}_1}
\boldsymbol{\Phi}_2\boldsymbol{\phi}_2\|_2 \nonumber\\
=&\max_{i\in\{l_1,\ldots,l_k\}}\|\mathbf{A}_i^T\mathcal{P}_{\boldsymbol{\Phi}_1}
\boldsymbol{\Phi}_2\boldsymbol{\phi}_2\|_2 \nonumber\\
\leq&
d\mu_{\text{B}}(K-k)\max_{i\in\{l_{k+1},\ldots,l_K\}}\|\mathbf{\tilde{x}}_{i}\|_2
\end{align}
where the last inequality comes from (\ref{appB:eq1}).


\bibliography{MyBibTex}
\bibliographystyle{IEEEtran}

\end{document}